\documentclass[letterpaper,11pt]{article}
\usepackage{geometry}
\usepackage{amsmath,color}
\usepackage{commath}
\usepackage{amsthm}
\usepackage{amssymb}
\usepackage{setspace,mathrsfs}
\usepackage{algorithm,natbib}
\usepackage{algorithmic}
\usepackage{url,amsmath,amsfonts,amssymb}
\usepackage{caption,subcaption,algorithm,graphicx}
\usepackage{listings,enumerate,color,relsize,multirow,rotating}

\def\bzero{\mathbf{0}}

\newtheorem{prop}{Proposition}

\newtheorem{remark}{Remark}

\newcommand{\trans}{^{\intercal}}

\def\bgamma{\boldsymbol{\gamma}}

\def\Isc{\mathcal{I}}

\def\X{\boldsymbol{X}}
\def\btheta{\boldsymbol{\theta}}
\def\x{\boldsymbol{x}}
\def\balpha{\boldsymbol{\alpha}}

\def\subpnorm{_{{\rm P},2}}
\def\bvarphi{\boldsymbol{\varphi}}
\def\bpsi{\boldsymbol{\psi}}
\def\bmu{\boldsymbol{\mu}}
\def\supmfk{^{[\text{-}k]}}
\def\supmfkj{^{[\text{-}k,\text{-}j]}}

\def\submfkj{_{\text{-}k,\text{-}j}}
\def\bC{\boldsymbol{C}}

\definecolor{darkred}{RGB}{150,50,50}
\definecolor{brown}{RGB}{250,100,100}
\definecolor{green}{RGB}{000,150,100}
\definecolor{purple}{RGB}{250,000,180}

\begin{document}

\title{A Note on Debiased/Double Machine Learning \\ Logistic Partially Linear Model}
\date{}
\author{Molei Liu$^1$}

\footnotetext[1]{Department of Biostatistics, Harvard Chan School of Public Health.}

\maketitle

\begin{abstract}
\noindent 
It is of particular interests in many application fields to draw doubly robust inference of a logistic partially linear model with the predictor specified as combination of a targeted low dimensional linear parametric function and a nuisance nonparametric function. In recent, \cite{tan2019doubly} proposed a simple and flexible doubly robust estimator for this purpose. 
They introduced the two nuisance models, i.e. nonparametric component in the logistic model and conditional mean of the exposure covariates given the other covariates and fixed response, and specified them as fixed dimensional parametric models. Their framework could be potentially extended to machine learning or high dimensional nuisance modelling exploited recently, e.g. in \cite{chernozhukov2016double,chernozhukov2018double} and \cite{smucler2019unifying,tan2020model}. 

Motivated by this, we derive the debiased/double machine learning logistic partially linear model in this note. For construction of the nuisance models, we separately consider the use of high dimensional sparse parametric models and general machine learning methods. By deriving certain moment equations to calibrate the first order bias of the nuisance models, we preserve a model double robustness property on high dimensional ultra-sparse nuisance models. We also discuss and compare the underlying assumption of our method with debiased LASSO \citep{van2014asymptotically}. To implement the machine learning proposal, we design a ``full model refitting" procedure that allows the use of any blackbox conditional mean estimation method in our framework. Under the machine learning setting, our method is rate doubly robust in a similar sense as \cite{chernozhukov2016double}.

\end{abstract}


\section{Introduction}

Consider a logistic partially linear model. Let $\{(Y_i,A_i,\X_i):i=1,2,\ldots,n\}$ be independent and identically distributed samples of $Y\in\{0,1\}$, $A\in\mathbb{R}$ and $\X\in\mathbb{R}^p$. Assume that
\begin{equation}
{\rm P}(Y=1\mid A,\X)={\rm expit}\{\beta_0 A+r_0(\X)\},
\label{model}
\end{equation}
where ${\rm expit}(\cdot)={\rm logit}^{-1}(\cdot)$, ${\rm logit}(a)=\log\{a/(1-a)\}$ and $r_0(\cdot)$ is an unknown nuisance function of $\X$. In a casual scenario with $A$ taken as the exposure (treatment) of interests, $Y$ being the observed outcome and $\X$ representing all the confounding variables, the parameter $\beta_0$ is of particular interests in that it measures the casual effect of $A$ on the potential outcome, in a scale of logarithmic odds ratio. And as the most common and natural way to characterize the casual model for a binary outcome, model (\ref{model}) is considered in extensive application fields like medical science, economy and political science. While for observational studies with unobserved confounding variables such as electronic health record (EHR) studies, model (\ref{model}) also plays an importatn role in studying the association between a phenotype (disease outcome) $Y$ and a key feature $A$, e.g. the diagnosis code for $Y$, conditional on patient profiles $\X$.

Our goal is to estimate and infer $\beta_0$ asymptotic normally at the rate $n^{-1/2}$. It has beend shown that if $\X$ is a scalar and $r_0(\cdot)$ is smooth, the semiparametric kernel or sieve regression \citep{severini1994quasi,lin2006semiparametric} works well for this purpose. However, when $\X$ is of relatively high dimensionality, these classic approaches have poor performance due to curse of dimensionality and one would specify $r_0(\x)$ in a parametric form: $r(\x)=\x\trans\bgamma$. To enhance the robustness to the potential misspecification of $r(\x)$, \cite{tan2019doubly} proposed a doubly robust estimating equation for $\beta$ including a parametric model $m(\x)=g(\x\trans\balpha)$ with a known link function $g(\cdot)$ for the conditional mean $m_0(\x)={\rm E}(A\mid Y=0,\X=\x)$:
\begin{equation}
\frac{1}{n}\sum_{i=1}^n\widehat\phi(\X_i)\left\{Y_ie^{-\beta A_i-\X_i\trans\widehat\bgamma}-(1-Y_i)\right\}\left\{A_i-g(\X_i\trans\widehat\balpha)\right\}=0,
\label{equ:dr:par}
\end{equation}
where $\widehat\phi(\x)$ is an estimation of some scalar nuisance function $\phi(\x)$ affecting the asymptotic efficiency of the estimator, and $\widehat\balpha$ and $\widehat\bgamma$ are two fixed dimensional nuisance model estimators. As demonstrated in \cite{tan2019doubly}, $\widehat\beta$ solved from (\ref{equ:dr:par}) is doubly robust in the sense that it is valid when either $r(\x)=\x\trans\bgamma$ is correctly specified for the nonparametric component in the logistic partial model, or $m(\x)=g(\x\trans\balpha)$ is correct for the conditional mean model $m_0(\x)={\rm E}(A\mid Y=0,\X=\x)$. It shows a novel doubly robustness property since prior to this, the doubly robust semiparametric estimation of odds ratio was built upon $p(A\mid \X,Y=0)$, the conditional density of $A$ given $\X$ and $Y=0$ \citep[e.g.]{yun2007semiparametric,tchetgen2010doubly}, requiring a stronger model assumption than (\ref{equ:dr:par}) for continuous $A$.

Nevertheless, \cite{tan2019doubly} focuses on fixed dimensional parametric nuisance models that are still prone to misspecification in practice. And their proposed framework is not readily applicable to the high dimensional \citep{athey2016approximate,chernozhukov2018double,smucler2019unifying,tan2020model} or general machine learning \citep{chernozhukov2016double} nuisance models frequently exploited in recent years. This is because for such nuisance models with higher complexity, simply using them to replace the fixed dimensional parametric models in (\ref{equ:dr:par}) incurs excessive fitting bias and does not guarantee the desirable property of $\widehat\beta$. In addition, estimating $r_0(\x)$ with arbitrary machine learning algorithms (of conditional mean) is not flexible because it is linked with the response through a nonlinear logit function. In this note, we handle these challenges and fill the gap by deriving the extensions of (\ref{equ:dr:par}) to accommodate high dimensional sparse nuisance models or general machine learning nuisance models separately.

For the high dimensional sparse model setting, i.e. $p\gg n$ and the two nuisance components are specified as parametric models with sparse coefficients, we realize bias reduction with respect to regularization errors of the nuisance estimators through certain dantzig moment equations of $\X$. Under the ultra-sparsity assumption of the nuisance models, our estimator preserves the same model double robustness property as the fixed ($p\ll n$) dimensional nuisance models. Compared with the debiased (desparsified) LASSO estimator for logistic model \citep{van2014asymptotically,jankova2016confidence}, we find our model sparsity assumption is more reasonable and explainable while debiased LASSO is being criticized on requiring the inverse information matrix to be sparse, a generally unverifiable technical condition \citep{xia2020revisit}.

Under the general machine learning framework, 
our approach allows for the use of any blackbox learning algorithm for condition mean estimation as in \cite{chernozhukov2016double}. Unlike the partially linear model considered in their paper, this generality is not readily achievable on logistic model due to its non-linear link function. We propose a easy-to-implement ``full model refitting" procedure to handle this problem and make implementation of learning algorithms flexible in our framework. Similar to \cite{chernozhukov2016double}, we discuss the rate double robustness property of the proposed estimator assuming that the machine learning estimation of the two nuisance models approaches the true models at certain geometric rates.

\section{Some preliminary derivation}\label{sec:moti}

Before introducing the specific methods in Section \ref{sec:method}, we first present a (simplified) generalization of the doubly robust estimating equation (\ref{equ:dr:par}) and derive its first and second order error decomposition, which plays a central role in motivating and guiding our method construction and theoretical analysis. Suppose the nuisance models $r_0(\x)$ and $m_0(\x)$ are estimated by $\widehat r(\x)$ and $\widehat m(\x)$ that approach some limiting functions $\bar r(\x)$ and $\bar m(\x)$. Motivated by (\ref{equ:dr:par}), we consider
\begin{equation}
\frac{1}{n}\sum_{i=1}^n\left\{Y_ie^{-\beta A_i}-(1-Y_i)e^{\widehat r(\X_i)}\right\}\left\{A_i-\widehat m(\X_i)\right\}=0,
\label{equ:dr:general}    
\end{equation}
and denote its solution as $\widehat\beta$. Compared with (\ref{equ:dr:par}), we omit here a multiplicative factor $e^{-\widehat r(\X_i)}\widehat\phi(\X_i)$ that will only affect asymptotic variance of the estimator, to simplify the formation so that the attention would not be distracted from our main idea. And we shall comment on the incorporation of this nuisance function with our framework in Section \ref{sec:eff}.

Concerning the error depending on $\widehat r(\cdot)$ and $\widehat m(\cdot)$, we decompose equation (\ref{equ:dr:general}) as follows:
\begin{equation}
\begin{split}
&\frac{1}{n}\sum_{i=1}^n\left\{Y_ie^{-\beta A_i}-(1-Y_i)e^{\widehat r(\X_i)}\right\}\left\{A_i-\widehat m(\X_i)\right\}\\
=&\frac{1}{n}\sum_{i=1}^nh(Y_i,A_i,\X_i;\bar r(\cdot),\bar m(\cdot))-\frac{1}{n}\sum_{i=1}^n\left\{Y_ie^{-\beta A_i}-(1-Y_i)e^{\bar r(\X_i)}\right\}\left\{\widehat m(\X_i)-\bar m(\X_i)\right\}\\
&-\frac{1}{n}\sum_{i=1}^n(1-Y_i)e^{\bar r(\X_i)}\left\{\widehat r(\X_i)-\bar r(\X_i)\right\}\left\{A_i-\bar m(\X_i)\right\}\\
&+O_p\left(\|\widehat r(\X)-\bar r(\X)\|_{{\rm P},2}^2+\|\widehat m(\X)-\bar m(\X)\|_{{\rm P},2}^2\right)+o_p(1/\sqrt{n}),
\end{split}  
\label{equ:expand}
\end{equation}
where we denote by $h(Y_i,A_i,\X_i;\bar r(\cdot),\bar m(\cdot))=\{Y_ie^{-\beta A_i}-(1-Y_i)e^{\bar r(\X_i)}\}\{A_i-\bar m(\X_i)\}$, define $\|f(\X)\|\subpnorm={\rm E}f^2(\X)$, and extract the second order terms (and beyond) as $\|\widehat r(\X)-\bar r(\X)\|_{{\rm P},2}^2+\|\widehat m(\X)-\bar m(\X)\|_{{\rm P},2}^2$ under certain mild regularity conditions.
When at least one nuisance model is correctly specified, i.e. $\bar r(\cdot)=r_0(\cdot)$ or $\bar m(\cdot)=m_0(\cdot)$ holds, we have 
\[
{\rm E}(1-Y)\{e^{\bar r(\X)}-e^{r_0(\X)}\}\{A-\bar m(\X)\}={\rm E}\left[\{e^{\bar r(\X)}-e^{r_0(\X)}\}\{A-\bar m(\X)\}\Big| Y=0,\X\right]=0,
\]
leading to ${\rm E}h(Y,A,\X;\bar r(\cdot),\bar m(\cdot))={\rm E}h(Y,A,\X;r_0(\cdot),\bar m(\cdot))$ and $\beta_0$ solves ${\rm E}h(Y,A,\X;\bar r(\cdot),\bar m(\cdot))=0$. Similar to various existing work like \cite{chernozhukov2016double,chernozhukov2016locally,chernozhukov2018double} and \cite{tan2020model}, the root mean squared errors (rMSEs) of high dimensional parametric and machine learning methods, $\|\widehat r(\X)-\bar r(\X)\|_{{\rm P},2}$ and $\|\widehat m(\X)-\bar m(\X)\|_{{\rm P},2}$, are assumed to be $o_p(n^{-1/4})$ and consequently their impact is negligible asymptotically. Thus, it remains to remove the first order bias terms:
\begin{equation}
\begin{split}
&\Delta_m=\frac{1}{n}\sum_{i=1}^n\left\{Y_ie^{-\beta A_i}-(1-Y_i)e^{\bar r(\X_i)}\right\}\left\{\widehat m(\X_i)-\bar m(\X_i)\right\};\\
&\Delta_r=\frac{1}{n}\sum_{i=1}^n(1-Y_i)e^{\bar r(\X_i)}\left\{\widehat r(\X_i)-\bar r(\X_i)\right\}\left\{A_i-\bar m(\X_i)\right\},
\end{split}   
\label{equ:bias}
\end{equation}
In the low dimensional parametric case, these first order terms do not impact the asymptotic normality of $n^{1/2}(\widehat\beta-\beta_0)$ as the nuisance estimators themselves are asymptotic normal at rate $n^{-1/2}$. While for high dimensional and machine learning nuisance models, removal of them is not trivial due to the excessive fitting error of the nuisance models. And the non-negligible bias incurred by this is known as over-fitting (or first order) bias \citep{chernozhukov2016double}. In Section \ref{sec:method}, we shall derive the constructing procedure for complex nuisance models to remove $\Delta_m$ and $\Delta_r$ properly.


\section{Method}\label{sec:method}

\subsection{High dimensional sparse modeling}\label{sec:method:hd}

Consider the setting with $p\gg n$, $r(\x)=\x\trans\bgamma$ and $m(\x)=g(\x\trans\balpha)$ where $g(\cdot)$ is a monotone link function with derivative $g'(\cdot)$. We derive the constructing procedure for high dimensional sparse nuisance models preserving a similar (model) doubly robustness property as \cite{tan2019doubly}. 

First, we obtain $\widetilde\bgamma$ as some initial estimators for $\bgamma$. Estimating procedure for $\widetilde\bgamma$ is quite flexible as it only needs to satisfy that $\widetilde\bgamma$ converges to some sparse limiting parameter $\bgamma^*$ equaling to the true model parameter $\bgamma_0$ when the nuisance model $r(\x)=\x\trans\bgamma$ is correct. Motivated by Section \ref{sec:moti}, we propose to obtain $\widehat\balpha$ by solving the dantzig moment equation:
\begin{equation}
{\rm min}_{\balpha\in\mathbb{R}^p}\|\balpha\|_1\quad{\rm s.t}\quad\left\|n^{-1}\sum_{i=1}^n(1-Y_i)e^{\widetilde\bgamma\trans\X_i}\left\{A_i-g(\X_i\trans\balpha)\right\}\X_i\right\|_{\infty}\leq\lambda_{\alpha},
\label{equ:dant:m}    
\end{equation}
where $\lambda_{\alpha}$ is a tuning parameter controlling the regularization bias. Then we solve the nuisance $\bgamma$ and the target parameter $\beta$ jointly from:
\begin{equation}
\begin{split}
{\rm min}_{\beta\in\mathbb{R},\bgamma\in\mathbb{R}^p}\|\bgamma\|_1\quad{\rm s.t}\quad\left\|n^{-1}\sum_{i=1}^n\left\{Y_ie^{-\beta A_i}-(1-Y_i)e^{\X_i\trans\bgamma}\right\}g'(\X_i\trans\widehat\balpha)\X_i\right\|_{\infty}&\leq\lambda_{\gamma};\\
n^{-1}\sum_{i=1}^n\left\{Y_ie^{-\beta A_i}-(1-Y_i)e^{\X_i\trans\bgamma}\right\}\left\{A_i-g(\X_i\trans\widehat\balpha)\right\}&=0,
\end{split}
\label{equ:dant:r}    
\end{equation}
Denote the solution of (\ref{equ:dant:r}) as $\widehat\beta$ and $\widehat\bgamma$. We demonstrate as follows that $\widehat\beta$ converges to $\beta_0$ at the parametric rate when at least one nuisance model is correct and both of them are ultra-sparse.

Similar to \cite{tan2020model}, the maximum-norm constraints (also known as Karus--Kuhn--Tucker condition) in (\ref{equ:dant:m}) and (\ref{equ:dant:r}) impose certain moment conditions to the nuisance parameters under potential model misspecification. We shall outline how this assists calibrating the first order bias terms in (\ref{equ:bias}). For simplification, we neglect some technical assumptions and analytical details that could be found in existing literature of high dimensional estimation and semiparametric inference\footnote{See \cite{candes2007dantzig}, \cite{bickel2009simultaneous}, \cite{buhlmann2011statistics} and \cite{negahban2012unified} for general theory of high dimensional regularized estimation. And see \cite{bradic2019sparsity}, \cite{smucler2019unifying} and \cite{tan2020model} for the theoretical framework of analyzing doubly robust estimator of the average treatment effect with high dimensional sparse nuisance models.}. Let $\bar\balpha$ and $\{\bar\bgamma,\bar\beta\}$ represent the limiting values of the solutions to (\ref{equ:dant:m}) and (\ref{equ:dant:r}) respectively, and $s$ be the maximum sparsity level of $\bar\balpha$, $\bar\bgamma$ and $\bgamma^*$. Following literature in high dimension statistics \citep{candes2007dantzig,bickel2009simultaneous,buhlmann2011statistics}, we assume that $\X_i$ is subgaussian with $O(1)$ scale. 
Then $\lambda_{\alpha}$ and $\lambda_{\gamma}$ are picked at the rate $\lambda=(\log p/n)^{1/2}$ and consequently, one could follow the analysis procedure in literatures like \cite{candes2007dantzig,bickel2009simultaneous,negahban2012unified} to derive that
\begin{equation}
\begin{split}
&\xi_1=\|\widetilde\bgamma-\bgamma^*\|_1+|\widehat\beta-\bar\beta|+\|\widehat\bgamma-\bar\bgamma\|_1+\|\widehat\balpha-\bar\balpha\|_1=O_p(s\lambda);\\
&\xi_2=\|\X\trans(\widetilde\bgamma-\bgamma^*)\|_{{\rm P},2}^2+\|A(\widehat\beta-\bar\beta)\|\subpnorm^2+\|\X\trans(\widehat\bgamma-\bar\bgamma)\|_{{\rm P},2}^2+\|\X\trans(\widehat\balpha-\bar\balpha)\|_{{\rm P},2}^2=O_p(s\lambda^2).
\end{split}    
\end{equation}

\begin{remark}
Note that (\ref{equ:dant:m}) involves the initial estimator $\widetilde\bgamma$ and (\ref{equ:dant:r}) involves the estimator $\widehat\balpha$ obtained beforehand, which requires some additional effort on removing their fitting errors in analyzing $\widehat\bgamma$ and $\widehat\balpha$, compared to the standard analysis procedures of dantzig selector. One could see \cite{bradic2019sparsity,smucler2019unifying,tan2020model} for a similar issue and to find relevant technical details being used for this purpose.
\end{remark}

Now consider the case when at least one nuisance model is correctly specified. Define that
\begin{align*}
\bvarphi(Y_i,A_i,\X_i;\beta,\bgamma,\balpha)=&\{Y_ie^{-\beta A_i}-(1-Y_i)e^{\X_i\trans\bgamma}\}g'(\X_i\trans\balpha)\X_i;\\
\bpsi(Y_i,A_i,\X_i;\bgamma,\balpha)=&(1-Y_i)e^{\bgamma\trans\X_i}\{A_i-g(\X_i\trans\balpha)\}\X_i.
\end{align*}
When $r(\x)$ is correctly specified, i.e. $r_0(\x)=\x\trans\bgamma_0$ for some $\bgamma_0$, we have $\bgamma^*=\bar\bgamma=\bgamma_0$. So it is satisfied that ${\rm E}\bvarphi(Y,A,\X;\bar\beta,\bar\bgamma,\bar\balpha)=\bzero$ by the correctness of $r(\x)$ and ${\rm E}\bpsi(Y,A,\X;\bar\bgamma,\bar\balpha)={\rm E}\bpsi(Y,A,\X;\bgamma^*,\bar\balpha)=\bzero$ by the moment condition in (\ref{equ:dant:m}). When $m(\x)$ is correct: $m_0(\x)=g(\x\trans\balpha_0)$ and $\bar\balpha=\balpha_0$, we have ${\rm E}\bpsi(Y,A,\X;\bar\bgamma,\bar\balpha)=\bzero$ and corresponding to the constraint in the first row of (\ref{equ:dant:r}), it holds that ${\rm E}\bvarphi(Y,A,\X;\bar\beta,\bar\bgamma,\bar\balpha)=\bzero$. In addition, $\bar\beta=\beta_0$ in both cases, according to the moment equation in the second row of (\ref{equ:dant:r}) and the discussion in Section \ref{sec:moti}. Combining these with the subgaussianity of the covariates, the two bias term in (\ref{equ:bias}) can be controlled through
\begin{equation*}
\begin{split}
\Delta_m\leq& \left\|n^{-1}\sum_{i=1}^n\bvarphi(Y_i,A_i,\X_i;\bar\beta,\bar\bgamma,\bar\balpha)\right\|_{\infty}\|\widehat\balpha-\bar\balpha\|_1+O_p(\xi_2)+o_p(1/\sqrt{n})\\
=&O_p\{(\log p/n)^{1/2}\}O_p(s\lambda)+O_p(s\lambda^2)+o_p(1/\sqrt{n})=O_p(s\log p/n)+o_p(1/\sqrt{n});\\
\Delta_r\leq& \left\|n^{-1}\sum_{i=1}^n\bpsi(Y_i,A_i,\X_i;\bgamma^*,\bar\balpha)\right\|_{\infty}\|\widehat\bgamma-\bar\bgamma\|_1+O_p(\xi_2)+o_p(1/\sqrt{n})\\
=&O_p(s\log p/n)+o_p(1/\sqrt{n}),
\end{split}   
\end{equation*}
where the second order error terms can be extracted again into $O_p(\xi_2)+o_p(1/\sqrt{n})$. Thus, when $s=o(n^{1/2}/\log p)$, $\Delta_m$ and $\Delta_r$ are below the parametric rate and consequently, equation (\ref{equ:expand}) is asymptotically equivalent with $n^{-1}\sum_{i=1}^nh(Y_i,A_i,\X_i;\bar r(\cdot),\bar m(\cdot))=0$ and $n^{1/2}(\widehat\beta-\beta_0)$ weakly converges to normal distribution with mean $0$ under some mild regularity conditions.

Though extensively studied and used in recent years, debiased LASSO \citep{zhang2014confidence,javanmard2014confidence,van2014asymptotically} has been criticized on that under the generalized linear model setting, its sparsity condition imposed on inverse of the information matrix, is not explainable and justifiable, leading to a subpar performance in practice \citep{xia2020revisit}. Interestingly, we find the model and sparsity assumption of our method is more reasonable than debiased LASSO and present a simple comparison of these two approaches in Remark \ref{rem:2}.

\begin{remark}
Assume the logistic model ${\rm P}(Y=1\mid A,\X)={\rm expit}\{\beta_0 A+\X\trans\bgamma_0\}$ is correctly specified. Let its expected information matrix be $\boldsymbol{\Sigma}_{\beta_0,\bgamma_0}={\rm E}[{\rm expit}'\{\beta_0 A+\X\trans\bgamma_0\}(A,\X\trans)\trans(A,\X\trans)]$, $\boldsymbol{\Theta}_{\beta_0,\bgamma_0}=\boldsymbol{\Sigma}_{\beta_0,\bgamma_0}^{-1}$ and $\btheta_{\beta_0}$ be the first column of $\boldsymbol{\Theta}_{\beta_0,\bgamma_0}$. In \cite{van2014asymptotically,jankova2016confidence}, asymptotic normality of the debiased logistic LASSO estimator for $\beta_0$ is built under the sparsity assumption: $\|\btheta_{\beta_0}\|_0=o(\{n/(\log p)^2\}^{1/3})$. Using cross-fitting to estimate $\boldsymbol{\Sigma}_{\beta_0,\bgamma_0}$, recent work like \cite{ma2020global} and \cite{liu2020integrative} has relaxed this condition to $\|\btheta_{\beta_0}\|_0=o(n^{1/2}/\log p)$, or replaced it with approximate sparsity assumptions such as $\|\btheta_{\beta_0}\|_1=O(1)$. However, in the presense of the weight ${\rm expit}'\{\beta_0 A+\X\trans\bgamma_0\}$ in $\boldsymbol{\Sigma}_{\beta_0,\bgamma_0}$, neither of these assumptions are explainable nor they are justifiable for the most common gaussian design \citep{xia2020revisit}. 

In comparison, we require that ${\rm E}(A\mid Y=0,\X=\x)=g(\X_i\trans\balpha_0)$ with $\|\balpha_0\|_0=o(n^{1/2}/\log p)$. This assumption has two advantages. First, it accommodates nonlinear link function $g(\cdot)$, which could make the assumption more reasonable for a categorical $A$. Second, our assumption is imposed on a conditional model directly so it is more explainable than debiased LASSO. As a simple example, consider a practically useful conditional gaussian model: $(A,\X\trans)\trans\mid \{Y=j\}\sim \mathcal{N}(\bmu_j,\boldsymbol{\Sigma})$ for $j=0,1$. Then we have $r_0(\x)=\x\trans\bgamma_0$ where $\bgamma_0=\Sigma^{-1}(\bmu_1-\bmu_0)$ and $A\mid \X,Y=0$ follows a gaussian linear model: $m_0(\x)=\x\trans\balpha_0$ with $\balpha_0$ determined by $\Sigma^{-1}$. Thus, our model sparsity assumptions on the nuisance coefficients $\balpha_0$ and $\bgamma_0$ are actually imposed on the data generation parameters $\Sigma^{-1}$ and $\bmu_1-\bmu_0$, which are more explainable and verifiable in practice.

\label{rem:2}
\end{remark}

\subsection{General machine learning modeling}\label{sec:method:ml}

In this section, we turn to a more general machine learning setting that any learning algorithms of conditional mean could be potentially applied to estimate $r_0(\cdot)$ and $m_0(\cdot)$. Similar to \cite{chernozhukov2016double}, we randomly split the $n$ samples into $K=O(1)$ folds: $\Isc_1,\Isc_2,\ldots,\Isc_K$ of equal size, to help remove the over-fitting bias. Denote by $\Isc_{\text{-}k}=\{1,\ldots,n\}\setminus\Isc_k$ and we replace estimating equation (\ref{equ:dr:general}) by the cross-fitted
\begin{equation}
\frac{1}{n}\sum_{k=1}^K\sum_{i\in\Isc_k}\left\{Y_ie^{-\beta A_i}-(1-Y_i)e^{\widehat r\supmfk(\X_i)}\right\}\left\{A_i-\widehat m\supmfk(\X_i)\right\}=0,
\label{equ:dr:cross}     
\end{equation}
where $\widehat r\supmfk(\cdot)$ and $m\supmfk(\cdot)$ are two machine learning estimators obtained with the samples in $\Isc_{\text{-}k}$ and converging to the true models $r_0(\cdot)$ and $m_0(\cdot)$. 
We outline as follows a strategy utilizing an arbitrary (supervised) learning algorithm to estimate the nuisance models.


Suppose there is a blackbox procedure $\mathscr{L}(R_i,\bC_i;\Isc)$ that inputs samples from $\Isc\subseteq\{1,2,\ldots,n\}$ with some response $R_i$ and covariates $\bC_i$ and outputs an estimation of ${\rm E}[R_i\mid\bC_i,i\in\Isc]$.\footnote{Without purposed modification, the natural form of most contemporary supervised learning methods, e.g. random forest, support vector machine and neural network, can be conceptualized in this way since their goal is prediction for a continuous response and classification for a categorical response.} Estimator of $m_0(\cdot)$ can be obtained by $\widehat m\supmfk(\cdot)=\mathscr{L}(A_i,\X_i;\Isc_{\text{-}k}\cap\{i:Y_i=0\})$. 
Unlike the partial linear setting in \cite{chernozhukov2016double}, estimating $r_0(\cdot)$ using $\mathscr{L}$ is more sophisticated since it is defined through an unextractable semiparametric form ${\rm P}(Y=1\mid A,\X)={\rm expit}\{\beta_0 A+r_0(\X)\}$. 
Certainly, one could modify some existing machine learning approaches like neural network\footnote{By setting the last layer of the neural network to be the combination of a complex network of $\X$ and a linear function of $A$ and linking it with the outcome through an expit link.} to accommodate this form. However, such modification is not readily available in general and even there was some way to adapt $\mathscr{L}$ to this semiparametric form, it would typically require additional human efforts on its implementing and validating. Thus, we introduce a ``full model refitting" procedure basing on an arbitrary algorithm $\mathscr{L}$ to estimate $r_0(\cdot)$. Our method is motivated by a simple proposition:
\begin{prop}
Let the full model $M(A,\X)={\rm P}(Y=1\mid A,\X)={\rm expit}\{\beta_0 A+r_0(\X)\}$. One have:
\[
\beta_0={\rm argmin}_{\beta\in\mathbb{R}} {\rm E}\left[{\rm logit}\{M(A,\X)\}-\beta(A-{\rm E}[A|\X])\right]^2.
\]
\label{prop:1}
\end{prop}
\begin{proof}
For any $\beta\in\mathbb{R}$, we have
\begin{align*}
&{\rm E}\left[{\rm logit}\{M(A,\X)\}-\beta(A-{\rm E}[A|\X])\right]^2={\rm E}\left\{\beta_0 A+r_0(\X)-\beta(A-{\rm E}[A|\X])\right\}^2\\
=&{\rm E}\left\{(\beta_0-\beta)(A-{\rm E}[A|\X])+\eta(\X)\right\}^2=(\beta_0-\beta)^2{\rm E}(A-{\rm E}[A|\X])^2+{\rm E}\{\eta(\X)\}^2,
\end{align*}
where $\eta(\X)=r_0(\X)+{\rm E}[A|\X]$. Thus, $\beta_0$ minimizes ${\rm E}\left[{\rm logit}\{M(A,\X)\}-\beta(A-{\rm E}[A|\X])\right]^2$.
\end{proof}
Randomly each split $\Isc_{\text{-}k}$ into $K$ folds $\Isc_{\text{-}k,1},\ldots\Isc_{\text{-}k,K},$ of equal size. Motivated by Proposition \ref{prop:1}, we first estimate the ``full" model $M(A,\X)$ leaving out each $\Isc_{\text{-}k,\text{-}j}$ in $\Isc_{\text{-}k}$: 
\[
\widehat M\supmfkj(\cdot)=\mathscr{L}(Y_i,(A_i,\X_i\trans)\trans;\Isc\submfkj),
\]
and learn $a(\X)={\rm E}[A|\X]$ as $\widehat a\supmfkj(\cdot)=\mathscr{L}(A_i,\X_i;\Isc\submfkj)$. Then fit the (cross-fitted) least square regression to obtain:
\begin{equation}
\breve\beta\supmfk={\rm argmin}_{\beta\in\mathbb{R}}\frac{1}{|\Isc_{\text{-}k}|}\sum_{j=1}^K\sum_{i\in\Isc_{\text{-}k,j}}\left[{\rm logit}\left\{\widehat M\supmfkj(A_i,\X_i)\right\}-\beta\left\{A_i-\widehat a\supmfkj(\X_i)\right\}\right]^2.
\label{equ:est:init}
\end{equation}
We use cross-fitting in (\ref{equ:est:init}) to avoid the over-fitting of $\widehat M\supmfkj(\cdot)$ and $a\supmfkj(\cdot)$. Estimator of $r_0(\X_i)$ could then be given through $r_0(\X_i)={\rm logit}\{M(A_i,\X_i)\}-\beta_0 A_i$. Note that the empirical version of ${\rm logit}\{M(A_i,\X_i)\}-\beta_0 A_i$ typically involves $A_i$ due to the discrepancy of the true $\beta_0$ and $M(\cdot)$ from their empirical estimation, which can impede the removal of over-fitting bias terms in (\ref{equ:bias}) since $A-m_0(\X)$ is not orthogonal to the error dependent on $A$. So we instead use the conditional mean of ${\rm logit}\{M(A,\X)\}-\beta_0 A$ on $\X$ to estimate $r_0(\cdot)$. Let $W_i={\rm logit}\{\widehat M\supmfkj(A_i,\X_i)\}$ for each $i\in\Isc_{\text{-}k,j}$ and obtain $\widehat t\supmfk(\cdot)=\mathscr{L}(W_i,\X_i;\Isc_{\text{-}k})$, as a ``refitting" estimation of $t(\x):={\rm E}[{\rm logit}\{M(A,\X)\}|\X=\x]$. Then the estimator of $r_0(\cdot)$ is given by:
\[
\widehat r\supmfk(\x)=\widehat t\supmfk(\x)-\breve\beta\supmfk \widehat a\supmfk(\x),\quad\mbox{where}\quad\widehat a\supmfk(\x)=\frac{1}{K}\sum_{j=1}^K\widehat a\supmfkj(\x).
\]
Alternatively, one can also refit $r_0(\cdot)$ through
\[
\widehat r\supmfk(\cdot)=\log\left(\frac{\mathscr{L}(e^{-\breve\beta\supmfk A_i},\X_i;\Isc_{\text{-}k}\cap\{i:Y_i=1\})}{\mathscr{L}(1-Y_i,\X_i;\Isc_{\text{-}k})}\right),
\]
inspired by the moment condition sufficient to identify $r_0(\x)$:
\[
{\rm E}\left[Ye^{-\beta_0 A}-(1-Y)e^{r_0(\X)}\Big|\X\right]={\rm E}\left[e^{-\beta_0 A}\Big|\X,Y=1\right]-e^{r_0(\X)}{\rm E}\left[(1-Y)\big|\X\right]=0.
\]
At last, we outline the theoretical investigation on the estimator $\widehat\beta$ finally solved from (\ref{equ:dr:cross}). Similar to \cite{chernozhukov2016double}, the analysis relies on certain mild regularity conditions and the assumptions that (i) $\mathscr{L}$ outputs uniformly consistent estimators for the conditional mean models in \emph{all} learning objects it is implemented on; (ii) rMSEs of these estimators output by $\mathscr{L}$ are controlled by $o_p(n^{-1/4})$.
\begin{remark}
The above described assumptions (i) and (ii) imply that $\mathscr{L}$ should perform similarly well on the learning objects with the covariates set as $\X$ or $(A,\X\trans)\trans$. Classic nonparametric regression approaches like kernel smoothing or sieve may not satisfy this because including one more covariate $A$ in the model could have substantial impact on their rate of convergence. Thus, we recommend using modern learning approaches that are more dimensionality-robust, such as random forest and neural network, in our framework. While the classic sieve or kernel construction for one-dimensional $\X$ in a type of ``plug-in" model has been well-studied in existing work like \cite{severini1994quasi,lin2006semiparametric}.
\end{remark}
Based on assumption (ii), we have that $\|\widehat m\supmfk(\X)-m_0(\X)\|\subpnorm=o_p(n^{-1/4})$,
\[
\breve\beta\supmfk=\frac{\sum_{i\in\Isc_{\text{-}k}}{\rm logit}\{M(A_i,\X_i)\}\{A_i-a(\X_i)\}}{\sum_{i\in\Isc_{\text{-}k}}\{A_i-a(\X_i)\}^2}+o_p(n^{-1/4})=\beta_0+O_p(n^{-1/2})+o_p(n^{-1/4}),
\]
and consequently $\|\widehat r\supmfk(\X)-r_0(\X)\|\subpnorm=o_p(n^{-1/4})$. Then the second order error presented in (the cross-fitted version of) (\ref{equ:expand}) is $o_p(n^{-1/2})$. Also, we remove the first order (over-fitting) bias defined by (\ref{equ:bias}) through assumption (i) and concentration, facilitated by the use of cross-fitting in (\ref{equ:dr:cross}), in the same spirit as \cite{chernozhukov2016double}. Combining these two results leads to that (\ref{equ:dr:cross}) is asymptotically equivalent with $n^{-1}\sum_{i=1}^nh(Y_i,A_i,\X_i;r_0(\cdot),m_0(\cdot))=0$ and thus $n^{1/2}(\widehat\beta-\beta_0)$ is asymptotically normal with mean $0$ under mild regularity conditions.


\subsection{Efficiency enhancing}\label{sec:eff}
We turn back to the construction of \cite{tan2019doubly}:
\begin{equation}
\frac{1}{n}\sum_{i=1}^n\widehat\phi(\X_i)e^{-\widehat r(\X_i)}\left\{Y_ie^{-\beta A_i}-(1-Y_i)e^{\widehat r(\X_i)}\right\}\left\{A_i-\widehat m(\X_i)\right\}=0,
\label{equ:dr:phi}    
\end{equation}
where $\widehat\phi(\X_i)$ is an empirical estimation of a (typically positive) nuisance function $\phi(\cdot)$ that depends on the nuisance models $\bar r(\cdot)$ and $\bar m(\cdot)$ and affects the asymptotic variance of $\widehat\beta$. \cite{tan2019doubly} proposed and studied two options for $\phi(\cdot)$ including:
\[
\phi_{\rm opt}(\X)=\frac{{\rm E}[\{A-\bar m(\X)\}^2|\X,Y=0]}{{\rm E}[\{A-\bar m(\X)\}^2/{\rm expit}\{\beta_0 A+\bar r(\X)\}|\X,Y=0]};\quad\phi_{\rm simp}(\X)={\rm expit}\{\bar r(\X)\}.
\]
It was shown that when both nuisance models are correctly specified, the estimator solved with the weight $\phi_{\rm opt}(\X)$ achieves the minimum asymptotic variance. Since calculation of $\phi_{\rm opt}(\X)$ involves numerical integration with respect to $\X$ given $Y=0$, it is sometimes inconvenient to implement. So \cite{tan2019doubly} also provides another simplified choice: $\phi_{\rm simp}(\X)$, obtained by evaluating $\phi_{\rm opt}(\X)$ at $\beta_0=0$.

Inclusion of the nuisance estimator $\widehat\phi(\X_i)e^{-\widehat r(\X_i)}$ incurs two challenges. First, it introduces additional bias terms. Second, formation of the first order bias $\Delta_m$ and $\Delta_r$ in (\ref{equ:bias}) alters. Correspondingly, we make some moderate modifications on methods described in Sections \ref{sec:method:hd} and \ref{sec:method:ml}.  We adopt again a cross-fitting strategy (for both the high dimensional parametric and machine learning settings) to obtain $\widehat\phi\supmfk(\X_i)e^{-\widehat r\supmfk(\X_i)}$ and plug-in it at $i\in\Isc_k$, for $k=1,2,\ldots,K$. Note that a function depending solely on $\X_i$ is orthogonal to $h(Y_i,A_i,\X_i;\bar r(\cdot),\bar m(\cdot))=\{Y_ie^{-\beta A_i}-(1-Y_i)e^{\bar r(\X_i)}\}\{A_i-\bar m(\X_i)\}$ when $\bar r(\cdot)=r_0(\cdot)$ or $\bar m(\cdot)=m_0(\cdot)$. Then as long as $\widehat\phi\supmfk(\X_i)e^{-\widehat r\supmfk(\X_i)}$ is consistent and at least one nuisance model is correct, we can remove the (cross-fitted) bias term:
\[
\frac{1}{n}\sum_{k=1}^K\sum_{i\in\Isc_k}\left\{\widehat\phi\supmfk(\X_i)e^{-\widehat r\supmfk(\X_i)}-\phi_i(\X_i)e^{\bar r(\X_i)}\right\}h(Y_i,A_i,\X_i;\bar r(\cdot),\bar m(\cdot))
\]
through concentration. Meanwhile, we note that the first order bias terms $\Delta_m$ and $\Delta_r$ defined in (\ref{equ:bias}) are weighted by $\widehat\phi\supmfk(\X_i)e^{-\widehat r\supmfk(\X_i)}$ under this efficiency enhancing construction.
So we naturally weight the moment equations (\ref{equ:dant:m}) and (\ref{equ:dant:r}) with $\widehat\phi\supmfk(\X_i)e^{-\widehat r\supmfk(\X_i)}$, as a modification of the high dimensional sparse modelling strategy. While in the machine learning scenario, both the nuisance estimators are supposed to approach the corresponding true models, i.e. $\bar r(\cdot)=r_0$ and $\bar m(\cdot)=m_0$ so there is no need to modify the way to obtain $\widehat m\supmfk(\cdot)$ and $\widehat r\supmfk(\cdot)$ in Section \ref{sec:method:ml}.

\section{Conclusion}

In this note, we extend the low dimensional parametric doubly robust approach for logistic partially linear model of \cite{tan2019doubly} to the settings where the nuisance models are estimated by the high dimensional sparse regression or general machine learning methods. For the high dimensional setting, we derive certain moment equations for the nuisance models to remove the first order bias. Also, we find the sparsity assumption of our approach is more explainable and reasonable than the ``sparse inverse information matrix" assumption used by debiased LASSO \citep{van2014asymptotically,jankova2016confidence}. For the general machine learning framework, we handle the non-linearity and ``unextractablility" issue of the logistic partial model using a ``full model refitting" procedure. This procedure is easy to implement and facilitates the use of arbitrary learning algorithms for the nuisance models in our framework. Meanwhile, it could be potentially extended to handle other similar structure issues like that of the partially linear $M$-estimator.

We also outline the key theoretical analysis procedures of our approaches and demonstrate the model double robustness of the high dimensional construction under ultra-sparsity assumptions and the rate double robustness of the machine learning setting. For the high dimensional setting, we note that our ultra-sparsity assumption, i.e. $s=o(n^{1/2}/\log p)$ on both nuisance models may be moderately relaxed through cross-fitting, inspired by \cite{smucler2019unifying}.

\section*{Acknowledgements}

The author thanks his advisor, Tianxi Cai, and collaborator, Yi Zhang, for helpful discussion and comments on this note.



\bibliographystyle{apalike}
\bibliography{library}

\newtheorem{thmx}{Theorem}
\renewcommand{\thethmx}{\Alph{thmx}}
\setcounter{lemma}{0}
\renewcommand{\thelemma}{A\arabic{lemma}}

\clearpage
\newpage
\setcounter{page}{1}
\appendix

\end{document}